\documentclass[oneside,reqno,english]{amsart}
\usepackage[T1]{fontenc}
\usepackage[latin9]{inputenc}
\usepackage{babel}
\usepackage{prettyref}
\usepackage{mathrsfs}
\usepackage{amstext}
\usepackage{amsthm}
\usepackage{amssymb}
\usepackage[pdfusetitle,
 bookmarks=true,bookmarksnumbered=false,bookmarksopen=false,
 breaklinks=false,pdfborder={0 0 0},pdfborderstyle={},backref=false,colorlinks=false]
 {hyperref}
\hypersetup{
 colorlinks=true,citecolor=blue,linkcolor=blue,linktocpage=true}

\makeatletter
\numberwithin{equation}{section}
\numberwithin{figure}{section}
\theoremstyle{plain}
\newtheorem{thm}{\protect\theoremname}[section]
\theoremstyle{plain}
\newtheorem{cor}[thm]{\protect\corollaryname}
\theoremstyle{plain}
\newtheorem{lem}[thm]{\protect\lemmaname}
\theoremstyle{definition}
\newtheorem{defn}[thm]{\protect\definitionname}
\theoremstyle{remark}
\newtheorem{rem}[thm]{\protect\remarkname}
\theoremstyle{plain}
\newtheorem{question}[thm]{\protect\questionname}
\theoremstyle{plain}
\newtheorem{prop}[thm]{\protect\propositionname}
\theoremstyle{remark}
\newtheorem*{rem*}{\protect\remarkname}
\theoremstyle{remark}
\newtheorem*{acknowledgement*}{\protect\acknowledgementname}

\newrefformat{cor}{Corollary~\ref{#1}}
\newrefformat{subsec}{Section~\ref{#1}}

\newcommand{\xyR}[1]{%
\xydef@\xymatrixrowsep@{#1}}

\makeatother

\providecommand{\acknowledgementname}{Acknowledgement}
\providecommand{\corollaryname}{Corollary}
\providecommand{\definitionname}{Definition}
\providecommand{\lemmaname}{Lemma}
\providecommand{\propositionname}{Proposition}
\providecommand{\questionname}{Question}
\providecommand{\remarkname}{Remark}
\providecommand{\theoremname}{Theorem}

\begin{document}
\title[]{Hilbert space-valued Gaussian processes, and quantum states}
\author{Palle E.T. Jorgensen}
\address{(Palle E.T. Jorgensen) Department of Mathematics, The University of
Iowa, Iowa City, IA 52242-1419, U.S.A.}
\email{palle-jorgensen@uiowa.edu}
\author{James Tian}
\address{(James F. Tian) Mathematical Reviews, 416 4th Street Ann Arbor, MI
48103-4816, U.S.A.}
\email{jft@ams.org}
\begin{abstract}
We offer new results and new directions in the study of operator-valued
kernels and their factorizations. Our approach provides both more
explicit realizations and new results, as well as new applications.
These include: (i) an explicit covariance analysis for Hilbert space-valued
Gaussian processes, (ii) optimization results for quantum gates (from
quantum information), (iii) new results for positive operator-valued
measures (POVMs), and (iv) a new approach/result in inverse problems
for quantum measurements.
\end{abstract}


\maketitle
\tableofcontents{}

\section{Introduction}

We present a new approach to inverse problems, combining classical
and quantum. Applications include new covariance analyses for Hilbert
space-valued Gaussian processes, and to a variety of dilation constructions.
We further present new results for optimization with quantum gates;
a new approach/result in inverse problems for quantum measurements,
and to inference with positive operator-valued measures (POVMs). 

With our starting point being that of operator-valued positive definite
kernels, our present paper then proceeds to offer a new two-fold framework
(i) and (ii) for the following problems: one is that of Hilbert space
and representations, and the other, new classes of Gaussian processes.
In more detail: We present: (i) a new approach to the interplay between
a variety of systems which allow for realizations in specific and
concrete Hilbert spaces $H$, on the one hand, and which at the same
time also admit extensions, typically called dilations, in \textquotedblleft larger\textquotedblright{}
Hilbert spaces, say $L$. (By \textquotedblleft dilation,\textquotedblright{}
we mean that the two systems are connected via an isometry $V$ from
$H$ to $L$, which intertwines the respective systems/representations.)
And for (ii), we make use of these constructions in our analysis of
Hilbert space-valued Gaussian processes. The constructions we present
in (i) are more explicit than those in earlier related analyses, such
as Gelfand Naimark Segal (GNS), from states to representations; Stinespring,
from completely positive maps to representations; or Halmos, from
powers contraction to unitary power dilations; and finally from positive
operator-valued measures (POVMs) to projection valued measures. See,
e.g., \cite{MR4114386,MR4581177,MR4482713,MR1465320,MR1469149,MR1740897,MR2945156,MR3217056,MR3250475,MR3394108}.
Our approach serves multiple purposes: one is to unify theories, and
the other is to make these constructions explicit.

Our new results include (i) an explicit covariance analysis for Hilbert
space-valued Gaussian processes (\prettyref{sec:4}), (ii) optimization
results for quantum gates (\prettyref{sec:5}), (iii) new results
for positive operator-valued measures (POVMs) (Sections \ref{sec:2}--\ref{sec:5}),
and (iv) a new approach/result in inverse problems for quantum measurements
(\prettyref{subsec:5-1}).

\textbf{Notation.} Throughout the paper, we use the physics convention
that inner products are linear in the second variable. Let $\left|a\left\rangle \right\langle b\right|$
denote Dirac's rank-1 operator, $c\mapsto a\left\langle b,c\right\rangle $.
$\mathcal{L}\left(H\right)$ denotes the algebra of all bounded linear
operators in a Hilbert space $H$. 

For a positive definite (p.d.) kernel $K:S\times S\rightarrow\mathbb{C}$,
let $H_{K}$ be the corresponding reproducing kernel Hilbert space
(RKHS). $H_{K}$ is the Hilbert completion of 
\begin{equation}
span_{\mathbb{C}}\left\{ K_{y}\left(\cdot\right):=K\left(\cdot,y\right)\mid y\in S\right\} \label{eq:A1}
\end{equation}
with respect to the inner product 
\begin{equation}
\left\langle \sum_{i}c_{i}K\left(\cdot,x_{i}\right),\sum_{i}d_{j}K\left(\cdot,y_{j}\right)\right\rangle _{H_{K}}:=\sum_{i}\sum_{j}\overline{c}_{i}d_{j}K\left(x_{i},x_{j}\right).\label{eq:A2}
\end{equation}
The following reproducing property holds: 
\begin{equation}
\varphi\left(x\right)=\left\langle K\left(\cdot,x\right),\varphi\right\rangle _{H_{K}},\forall x\in S,\:\forall\varphi\in H_{K}.\label{eq:A3}
\end{equation}

Any scalar-valued kernel $K$ as above is associated with a zero-mean
Gaussian process, where $K$ is the covariance: 
\begin{equation}
K\left(s,t\right)=\mathbb{E}\left[\overline{W_{s}}W_{t}\right],\label{eq:A4}
\end{equation}
and $W_{s}\sim N\left(0,K\left(s,s\right)\right)$. 

An $\mathcal{L}\left(H\right)$-valued kernel $K:S\times S\rightarrow\mathcal{L}\left(H\right)$
is p.d. if 
\begin{equation}
\sum_{i=1}^{n}\left\langle a_{i},K\left(s_{i},s_{j}\right)a_{j}\right\rangle _{H}\geq0\label{eq:A5}
\end{equation}
for all $\left(a_{i}\right)_{1}^{n}$ in $H$, and all $n\in\mathbb{N}$.

For the theory of reproducing kernels, the literature is extensive,
both pure and recent applications, we refer to the following, including
current, \cite{MR80878,MR31663,MR4302453,MR3526117,MR1305949,MR4690276,MR3700848,MR1126127,MR1200633,MR1329822,MR1473250,MR2223568,MR3402823}.

\section{\label{sec:2}Factorization of $\mathcal{L}\left(H\right)$-valued
p.d. kernels}

In this section, we present our universal construction in a general
framework. This will then, in adapted forms, go into a set of diverse
applications. The latter are addressed in subsequent sections.

In our analysis of kernels and factorizations we want to highlight
the following two distinctions: First is the distinction between the
traditional results which only yield existence of factorizations involving
abstract and unstructured completions, as opposed to arising as a
universal and concrete solution, realized directly from the given
operator valued kernel as an explicitly given space of functions.
The second is the distinction between universal, as opposed to case-by-case
analyses. In both senses, our results below fall in the category of
universal factorizations.
\begin{thm}[Universal Factorization]
\label{thm:b2} Let $K:S\times S\rightarrow\mathcal{L}\left(H\right)$
be a p.d. kernel. Then there exists a RKHS $H_{\tilde{K}}$, and a
family of operators $V_{s}:H\rightarrow H_{\tilde{K}}$, $s\in S$,
such that 
\begin{equation}
H_{\tilde{K}}=\overline{span}\left\{ V_{s}a:a\in H,s\in S\right\} 
\end{equation}
and 
\begin{equation}
K\left(s,t\right)=V_{s}^{*}V_{t}.\label{eq:b2}
\end{equation}

Conversely, if there is a Hilbert space $L$ and operators $V_{s}:H\rightarrow L$,
$s\in S$, such that 
\begin{equation}
L=\overline{span}\left\{ V_{s}a:a\in H,s\in S\right\} \label{eq:b3}
\end{equation}
and \eqref{eq:b2} holds, then $L\simeq H_{\tilde{K}}$. 
\end{thm}

\begin{proof}
First, assume that $K$ is of the form \eqref{eq:b2}. For all $\left(a_{i}\right)_{i=1}^{n}$
in $H$, and $\left(s_{i}\right)_{i=1}^{n}$ in $S$, 
\begin{align*}
\sum_{i,j}\left\langle a_{i},K\left(s_{i},s_{j}\right)a_{j}\right\rangle _{H} & =\sum_{i,j}\left\langle a_{i},V_{s_{i}}^{*}V_{s_{j}}a_{j}\right\rangle _{H}\\
 & =\sum_{i,j}\left\langle V_{s_{i}}a_{i},V_{s_{j}}a_{j}\right\rangle _{L}=\left\Vert \sum_{i}V_{s_{i}}a_{i}\right\Vert _{L}^{2}\geq0.
\end{align*}
Thus, $K$ is p.d.

Now, given a p.d. kernel $K:S\times S\rightarrow\mathcal{L}\left(H\right)$,
let 
\[
\tilde{K}:\left(S\times H\right)\times\left(S\times H\right)\rightarrow\mathbb{C}
\]
be the scalar-valued kernel defined by 
\[
\tilde{K}\left(\left(s,a\right),\left(t,b\right)\right):=\left\langle a,K\left(s,t\right)b\right\rangle _{H}.
\]
Then $\tilde{K}$ is p.d. and let $H_{\tilde{K}}$ be the associated
RKHS. 

Set $V_{s}:H\rightarrow H_{\tilde{K}}$ by 
\[
V_{s}a=\tilde{K}_{\left(s,a\right)}=\tilde{K}\left(\cdot,\left(s,a\right)\right):\left(S\times H\right)\rightarrow\mathbb{C},\quad\forall a\in H.
\]
Since 
\[
\left\langle V_{s}a,\tilde{K}\left(\cdot,\left(t,b\right)\right)\right\rangle _{H_{\tilde{K}}}=\left\langle a,K\left(s,t\right)b\right\rangle _{H},\quad\forall a,b\in H,\:\forall s,t\in S,
\]
it follows that 
\[
V_{s}^{*}\tilde{K}\left(\cdot,\left(t,b\right)\right)=K\left(s,t\right)b.
\]
One concludes that 
\[
K\left(s,t\right)=V_{s}^{*}V_{t}.
\]

Finally, suppose there is a Hilbert space $L$ and family of operators
$V_{s}:H\rightarrow L$, $s\in S$, such that \eqref{eq:b2}--\eqref{eq:b3}
hold. Then the mapping 
\[
\tilde{K}\left(\cdot,\left(s,a\right)\right)\longmapsto V_{s}a
\]
extends (by linearity and density) to a unique isometric isomorphism
from $H_{\tilde{K}}$ onto $L$, and so any such space $L$ is isomorphic
to the RKHS $H_{\tilde{K}}$.
\end{proof}
\begin{cor}
Any function $F\in H_{\tilde{K}}$ satisfies that 
\begin{equation}
F\left(s,a\right)=\left\langle a,V_{s}^{*}F\right\rangle _{H_{\tilde{K}}},\quad\left(s,a\right)\in S\times H.
\end{equation}
In particular, $F$ is conjugate liner in the second variable.
\end{cor}

\begin{proof}
This follows from the reproducing property (see \eqref{eq:A3}) of
$H_{\tilde{K}}$: 
\[
F\left(s,a\right)=\left\langle \tilde{K}\left(\cdot,\left(s,a\right)\right),F\right\rangle _{H_{\tilde{K}}}=\left\langle V_{s}a,F\right\rangle _{H_{\tilde{K}}}=\left\langle a,V_{s}^{*}F\right\rangle _{H}.
\]
\end{proof}
\begin{cor}
Let $\left(\varphi_{i}\right)$ be any orthonormal basis (ONB) in
$H_{\tilde{K}}$. The following operator identity holds:
\begin{equation}
K\left(s,t\right)=\sum_{i}\left|V_{s}^{*}\varphi_{i}\left\rangle \right\langle V_{t}^{*}\varphi_{i}\right|.\label{eq:b5}
\end{equation}
\end{cor}

\begin{proof}
Using the identity in $H_{\tilde{K}}$, 
\[
I_{H_{\tilde{K}}}=\sum_{i}\left|\varphi_{i}\left\rangle \right\langle \varphi_{i}\right|,
\]
we get that 
\[
K\left(s,t\right)=V_{s}^{*}V_{t}=V_{s}^{*}I_{H_{\tilde{K}}}V_{t}=\sum_{i}\left|V_{s}^{*}\varphi_{i}\left\rangle \right\langle V_{t}^{*}\varphi_{i}\right|.
\]
\end{proof}

\section{\label{sec:3}Relation to classical dilation theory}

The notion of dilation theory dates back to the fundamental Hilbert
space axioms, including to the study of such basic notions in operator
theory as contractive operators in Hilbert space $H$, of states defined
as positive linear functionals (on algebras of operators), and of
the notion and properties of Stinespring-completely positive operator-valued
maps. While the early advances were motivated primarily by pure operator
theory, and by representation theory, recent studies have served to
link much more directly to modern quantum information, see also \prettyref{sec:5}
below. The cast of characters in the general context was (and is)
perhaps focused on such areas of dilation theory as: dilations for
(i) subnormal operators, (ii) for operator valued measures, (iii)
for contraction operators, (iv) for operator spaces, and their role
as extensions in dilation theory, and (v) for commuting sets of operators,
and semigroups of completely positive maps. Early applications include
systems theory, prediction theory, and more recently quantum information.
The literature is vast, and here we call attention to \cite{MR2743416},
and the papers cited there. For the general theory of dilations of
irreversible evolutions in algebraic quantum theory we refer to the
two pioneering books \cite{MR489494,MR582649}, covering both the
theory and applications of completely positive maps, including semigroups.

We now turn to the role our above results (\prettyref{sec:2}) play
as tools in both classical and modern dilation theory; as well as
in sharpening the conclusions, and in presenting new directions.

Indeed, the setting in \prettyref{thm:b2} is quite general. It includes
and extends well known constructions in classical dilation theory.
Notably, passing from $K$ to a scalar-valued p.d. kernel $\tilde{K}$
offers a function-based approach to dilation theory, in contrast with
traditional abstract spaces of equivalence classes. It also allows
for direct evaluation and manipulation of functions, due to the RKHS
structure of the dilation space, as well as the associated $H$-valued
Gaussian processes (\prettyref{sec:4}), making it especially relevant
for practical applications. 

For instance, consider a completely positive map $\varphi$ from a
$C^{*}$-algebra $\mathfrak{A}$ to the algebra of bounded operators
on a Hilbert space $H$. The Stinespring Dilation Theorem \cite{MR69403}
states that there exists a Hilbert space $L$ (potentially larger
than $H$), a representation $\pi:\mathfrak{A}\rightarrow\mathcal{L}\left(L\right)$,
and a bounded operator $V:H\rightarrow L$, such that 
\begin{equation}
\varphi\left(A\right)=V^{*}\pi\left(A\right)V,\quad A\in\mathfrak{A}.\label{eq:b6}
\end{equation}

Moreover, if $\varphi$ is unital, then $V$ is isometric and $H$
can be identified as a subspace of $L$. The space $L$ is minimal
if the span of $\left\{ \pi\left(A\right)Va:a\in H,A\in\mathcal{L}\left(H\right)\right\} $
is dense in $L$, in which case $L$ is unique up to unitary equivalence. 

A key step in the proof of Stinespring's theorem is to introduce a
sesquilinear form on the algebraic tensor product $\mathfrak{A}\otimes H$,
\[
\left\langle \sum_{i}A_{i}\otimes a_{i},\sum_{j}B_{j}\otimes b_{j}\right\rangle _{L}:=\sum_{i,j}\left\langle a_{i},\varphi\left(A^{*}B\right)b_{j}\right\rangle _{H},
\]
and then setting $L$ to be the corresponding Hilbert completion,
modulo elements with zero norm, i.e., elements such that $\left\Vert \sum A_{i}\otimes a_{i}\right\Vert _{L}=0$.
The space $L$ is then an abstract Hilbert space of equivalence classes. 

To compare with \prettyref{thm:b2}, one may consider the $\mathcal{L}\left(H\right)$-valued
p.d. kernel 
\[
K\left(A,B\right)=\varphi\left(A^{*}B\right)
\]
defined on $\mathfrak{A}\times\mathfrak{A}$, and define the scalar-valued
p.d. kernel $\tilde{K}:\left(\mathfrak{A}\times H\right)\times\left(\mathfrak{A}\times H\right)\rightarrow\mathbb{C}$
by 
\[
\tilde{K}\left(\left(A,a\right),\left(B,b\right)\right)=\left\langle a,\varphi\left(A^{*}B\right)b\right\rangle _{H}.
\]

\begin{cor}
The space $L$ is then the RKHS $H_{\tilde{K}}$, consisting of functions
$F$ on $\mathfrak{A}\times H$, with the reproducing property: 
\[
F\left(A,a\right)=\left\langle \tilde{K}\left(\cdot,\left(A,a\right)\right),F\right\rangle _{H_{\tilde{K}}}.
\]
Note that $H_{\tilde{K}}$ is a function space on $\mathfrak{A}\times H$,
as opposed to an abstract space of equivalence classes. 
\end{cor}

The operator $V$, and the representation $\pi$, are also explicit:
\begin{lem}
Set $V:H\rightarrow H_{\tilde{K}}$ by 
\[
Va=\tilde{K}\left(\cdot,\left(I,a\right)\right).
\]
The adjoint $V^{*}:H_{\tilde{K}}\rightarrow H$ is given by 
\[
V^{*}\tilde{K}\left(\cdot,\left(B,b\right)\right)=\varphi\left(B\right)b.
\]
Define $\pi:\mathfrak{A}\rightarrow\mathcal{L}\left(H_{\tilde{K}}\right)$
as 
\[
\pi\left(A\right)\tilde{K}\left(\cdot,\left(B,b\right)\right)=\tilde{K}\left(\cdot,\left(AB,b\right)\right).
\]
Then \eqref{eq:b6} holds. Indeed, since, by definition,
\[
span\left\{ K\left(\cdot,\left(A,a\right)\right):A\in\mathfrak{A},a\in H\right\} 
\]
is dense in $H_{\tilde{K}}$, $\left(\mathfrak{A},\pi,V,H_{\tilde{K}}\right)$
forms a minimal Stinespring dilation of $\varphi$. 
\end{lem}

\begin{proof}
We verify the adjoint $V^{*}$ is as stated. Since the span of 
\[
\left\{ \tilde{K}\left(\cdot,\left(A,a\right)\right):A\in\mathfrak{A},a\in H\right\} 
\]
is dense in $H_{\tilde{K}}$, it is enough to check that 
\begin{align*}
\left\langle Va,\tilde{K}\left(\cdot,\left(B,b\right)\right)\right\rangle _{H_{\tilde{K}}} & =\left\langle \tilde{K}\left(\cdot,\left(I,a\right)\right),\tilde{K}\left(\cdot,\left(B,b\right)\right)\right\rangle _{H_{\tilde{K}}}\\
 & =\left\langle a,\varphi\left(B\right)b\right\rangle _{H}
\end{align*}
which holds for all $a,b\in H$, and all $B\in\mathcal{L}\left(H\right)$.
This means $V^{*}\tilde{K}\left(\cdot,\left(B,b\right)\right)=\varphi\left(B\right)b$,
for all $B\in\mathfrak{A}$, and all $b\in H$. 

Then, 
\begin{align*}
V^{*}\pi\left(A\right)Va & =V^{*}\pi\left(A\right)\tilde{K}\left(\cdot,\left(I,a\right)\right)\\
 & =V^{*}\tilde{K}\left(\cdot,\left(A,a\right)\right)=\varphi\left(A\right)a
\end{align*}
which is \eqref{eq:b6}.

Additionally, if $\varphi$ is unital, then 
\begin{align*}
\left\Vert Va\right\Vert _{H_{\tilde{K}}}^{2} & =\left\langle \tilde{K}\left(\cdot,\left(I,a\right)\right),\tilde{K}\left(\cdot,\left(I,a\right)\right)\right\rangle _{H_{\tilde{K}}}\\
 & =\left\langle a,K\left(I,I\right)a\right\rangle _{H}=\left\langle a,\varphi\left(I\right)a\right\rangle _{H}=\left\langle a,a\right\rangle _{H}=\left\Vert a\right\Vert _{H}^{2},
\end{align*}
so that $V$ is isometric. 
\end{proof}

\section{\label{sec:4}$H$-valued Gaussian processes}

Since the foundational work of Ito, the stochastic analysis of Gaussian
processes has relied heavily on Hilbert space and associated operators.
For reference, see works such as \cite{MR45307,MR3402823,MR4302453}
along with the early literature cited therein. In this discussion,
we highlight two distinct scenarios: (i) scalar-valued Gaussian processes
indexed by a Hilbert space, often through an Ito-isometry; and (ii)
Gaussian processes that take values in a Hilbert space themselves,
see e.g., \cite{MR4641110,MR4414825,MR4101087,MR4073554,MR3940383,doi:10.1142/S0219025723500200}.
Our current focus is on the latter. This setting offers the advantage
of enabling more flexible representations of covariance structures,
which are increasingly relevant in the analysis of large data sets.
We now proceed with the construction details.

Let $S$ be a set, and $H$ a Hilbert space. We say $\left\{ W_{s}\right\} _{s\in S}$
is an $H$-valued Gaussian process if, for all $a,b\in H$, and all
$s,t\in S$, 
\begin{equation}
\mathbb{E}\left[\left\langle a,W_{s}\right\rangle _{H}\left\langle W_{t},b\right\rangle _{H}\right]=C_{a,b}\left(s,t\right),
\end{equation}
where 
\begin{equation}
\left\langle W_{s},a\right\rangle _{H}\sim N\left(0,C_{a,a}\left(s,s\right)\right),
\end{equation}
i.e., mean zero Gaussian, with variance $C_{a,a}\left(s,s\right)$. 
\begin{thm}
Every operator-valued p.d. kernel $K:S\times S\rightarrow\mathcal{L}\left(H\right)$
is associated with an $H$-valued Gaussian process $\left\{ W_{s}\right\} _{s\in S}$,
such that 
\[
\mathbb{E}\left[\left\langle a,W_{s}\right\rangle _{H}\left\langle W_{t},b\right\rangle _{H}\right]=\left\langle a,K\left(s,t\right)b\right\rangle _{H}.
\]
for all $s,t\in S$, and all $a,b\in H$. 

Conversely, any $H$-valued Gaussian process $\left\{ W_{s}\right\} _{s\in S}$
is obtained from such a p.d. kernel. 
\end{thm}

\begin{proof}
Suppose $\left\{ W_{s}\right\} _{s\in S}$ is an $H$-valued (mean
zero) Gaussian process realized in a probability space $\left(\Omega,\mathbb{P}\right)$,
with expectation $\mathbb{E}\left(\cdot\right)=\int_{\Omega}\left(\cdot\right)d\mathbb{P}$. 

Fix $s,t\in S$. For all $a,b\in H$, 
\[
\left(a,b\right)\longmapsto\mathbb{E}\left[\left\langle a,W_{s}\right\rangle _{H}\left\langle W_{t},b\right\rangle _{H}\right]
\]
is a bounded sesquilinear form on $H$, and so 
\[
\mathbb{E}\left[\left\langle a,W_{s}\right\rangle _{H}\left\langle W_{t},b\right\rangle _{H}\right]=\left\langle a,K\left(s,t\right)b\right\rangle _{H}
\]
for some operator $K\left(s,t\right)\in\mathcal{L}\left(H\right)$.
Also, 
\[
K:S\times S\rightarrow\mathcal{L}\left(H\right)
\]
is positive definite, since 
\begin{align*}
\sum_{i,j=1}^{n}\left\langle a_{i},K\left(s_{i},s_{j}\right)a_{j}\right\rangle _{H} & =\sum_{i,j=1}^{n}\mathbb{E}\left[\left\langle a_{i},W_{s_{i}}\right\rangle _{H}\left\langle W_{s_{j}},a_{j}\right\rangle _{H}\right]\\
 & =\mathbb{E}\left[\left|\sum_{i=1}^{n}\left\langle W_{s_{i}},a_{i}\right\rangle _{H}\right|^{2}\right]\geq0
\end{align*}
for all $a_{1},\cdots,a_{n}\in H$ and $s_{1},\cdots,s_{n}\in S$. 

Conversely, given $K:S\times S\rightarrow\mathcal{L}\left(H\right)$
p.d., introduce the scalar-valued p.d. kernel $\tilde{K}:\left(S\times H\right)\times\left(S\times H\right)\rightarrow\mathbb{C}$,
\[
\tilde{K}\left(\left(s,a\right),\left(t,b\right)\right)=\left\langle a,K\left(s,t\right)b\right\rangle _{H}.
\]
Let $H_{\tilde{K}}$ be the RKHS of $\tilde{K}$, and let 
\[
\left\{ W_{\left(s,a\right)}:\left(s,a\right)\in S\times H\right\} 
\]
be the associated Gaussian process (see \eqref{eq:A4}, and also \prettyref{thm:c3}
for a realization). That is, 
\[
\tilde{K}\left(\left(s,a\right),\left(t,b\right)\right)=\left\langle a,K\left(s,t\right)b\right\rangle _{H}=\mathbb{E}\left[\overline{W_{\left(s,a\right)}}W_{\left(t,b\right)}\right].
\]

Let $V_{s}:H\rightarrow H_{\tilde{K}}$ be given by $V_{s}a=\tilde{K}\left(\cdot,\left(s,a\right)\right)$.
Using the isometry 
\[
H_{\tilde{K}}\ni V_{s}a=\tilde{K}\left(\cdot,\left(s,a\right)\right)\longrightarrow W_{\left(s,a\right)}\in L^{2}\left(\mathbb{P}\right)
\]
it follows that $H\ni a\rightarrow W_{\left(s,a\right)}$ is linear.
The assertion follows. 
\end{proof}
\begin{defn}
Consider the $H$-valued Gaussian process $W_{s}:\Omega\rightarrow H$,
\begin{equation}
W_{t}=\sum_{i}\left(V_{t}^{*}\varphi_{i}\right)Z_{i},\label{eq:d8}
\end{equation}
where $\left(\varphi_{i}\right)$ is an ONB in $H_{\tilde{K}}$. Following
standard conventions, here $\left\{ Z_{i}\right\} $ refers to a choice
of an independent identically distributed (i.i.d.) system of standard
scalar Gaussians $N(0,1)$ random variables, and with an index matching
the choice of ONB.
\end{defn}

\begin{thm}
\label{thm:c3}We have 
\begin{equation}
\mathbb{E}\left[\left\langle a,W_{s}\right\rangle _{H}\left\langle W_{t},b\right\rangle _{H}\right]=\left\langle a,K\left(s,t\right)b\right\rangle _{H}.\label{eq:d9}
\end{equation}
\end{thm}

\begin{proof}
Note that $\mathbb{E}\left[Z_{i}Z_{j}\right]=\delta_{i,j}$. From
this, we get
\begin{align*}
\text{LHS}_{\left(\ref{eq:d9}\right)} & =\mathbb{E}\left[\left\langle a,W_{s}\right\rangle _{H}\left\langle W_{t},b\right\rangle _{H}\right]\\
 & =\sum_{i}\sum_{j}\left\langle a,V_{s}^{*}\varphi_{i}\right\rangle \left\langle V_{t}^{*}\varphi_{j},b\right\rangle \mathbb{E}\left[Z_{i}Z_{j}\right]\\
 & =\sum_{i}\left\langle V_{s}a,\varphi_{i}\right\rangle _{H_{\tilde{K}}}\left\langle \varphi_{i},V_{t}b\right\rangle _{H_{\tilde{K}}}\\
 & =\left\langle V_{s}a,V_{t}b\right\rangle _{H_{\tilde{K}}}=\left\langle a,V_{s}^{*}V_{t}b\right\rangle _{H}=\left\langle a,K\left(s,t\right)b\right\rangle _{H}.
\end{align*}
Also see \eqref{eq:b5}.
\end{proof}
\begin{rem}
The Gaussian process $\left\{ W_{t}\right\} _{t\in S}$ in \eqref{eq:d8}
is well defined and possesses the stated properties. This is an application
of the central limit theorem to the choice $\left\{ Z_{i}\right\} $
of i.i.d. $N\left(0,1\right)$ Gaussians on the right-side of \eqref{eq:d8}.

Varying the choices of ONBs $\left(\varphi_{i}\right)$ and i.i.d.
$N\left(0,1\right)$ Gaussian random variables $\left(Z_{i}\right)$
will result in different Gaussian processes $\left\{ W_{t}\right\} _{t\in S}$,
but all will adhere to the covariance condition specified in \eqref{eq:d9}.
It is also important to note that if $\left(\varphi_{i}\right)$,
$\left(Z_{i}\right)$ are fixed, the resulting Gaussian process $\left\{ W_{t}\right\} _{t\in S}$
is uniquely determined by its first two moments.
\end{rem}

\section{\label{sec:5}Completeness, POVMs, quantum states, and quantum information}

In this section we apply our construction to the case of positive
operator-valued measures (POVMs). While POVMs have long played an
important role in operator theory, they have recently found applications
in quantum information \cite{MR3681182,MR3081872,MR3038250,MR3050559,MR2462580,MR1878924,MR1863141}.
In more detail: For a given Hilbert space $H$, and a measure space
$\Omega$, a positive operator-valued measure (POVM) on $\Omega$
is a sigma-additive measure whose values are positive semi-definite
operators on $H$. Thus the notion of POVM is a generalization of
projection-valued measures (PVMs). In quantum information, the latter
are called projection-measurements. (Here by a projection, we mean
selfadjoint projection $P$ in $H$.) Note that therefore quantum
measurements realized as POVMs are generalizations of quantum measurement
described by PVMs. By analogy, a POVM is to a PVM what a mixed state
is to a pure state. Recall that in quantum theory, mixed states specify
those quantum states occurring in subsystems of larger systems. Hence
the study of purification of quantum states. The POVMs describe the
effect on a subsystem of a projection measurement performed on a larger
system. In other words, POVMs are the most general kind of measurement
in quantum mechanics, and they in the field of quantum information.
Here quantum information refers the information of states of quantum
systems, thus leading to a variety of quantum information processing
techniques. The literature is extensive. For current references we
here mention only a sample, \cite{MR4678079,MR4612967,MR4736299,MR4732043,MR4710381,MR4705975,MR4509531,MR4280112,MR4256782}
and the papers cited there. Quantum information refers to both the
technical definition of von Neumann entropy as well as to the general
computational devices. Recently, quantum computing has become especially
active research area, serving as a key part of modern computation
and communication.

A main difference between classical probability, and its quantum counterpart
is played by the role of entangled quantum states, see e.g., \cite{MR4740579,MR4740573,MR4735667}.
The latter in turn have their origin with the pioneers. But in the
past decade, they have come to play a crucial role in modern quantum
computing. (We make a note on their origin: Erwin Schrödinger, in
studying the Einstein-Podolsky-Rosen (EPR) state of two particles
influencing each other instantly at a distance, introduced the concept
of entangled (quantum) states. He used the German term \textquotedblleft verschränkung\textquotedblright ,
which was subsequently translated into English as entanglement.)

In quantum information theory, a measurement channel is a POVM $P$,
defined on a measurable space $\left(\Omega,\mathscr{B}_{\Omega}\right)$,
acting in a Hilbert space $H$. The set of quantum states (density
operators) $\mathcal{S}_{1}^{1}\left(H\right)$ consists of positive
trace class operators $\rho$ with $tr\left(\rho\right)=1$. A measurement
channel $P$ is said to be information-complete if 
\[
\rho\mapsto tr\left(\rho P\left(\cdot\right)\right)
\]
from $\mathcal{S}_{1}^{1}\left(H\right)$ to the set of Borel probability
measures $\mathcal{M}_{1}\left(\Omega\right)$ on $\Omega$, is injective.
In other words, measurements uniquely determine states. In this context,
an $\mathcal{L}\left(H\right)$-valued p.d. kernel on $S\times S$
may be naturally viewed as a generalized quantum measurement channel. 

In what follows, let $\mathscr{K}\left(S,\mathcal{L}\left(H\right)\right)$
be the set of all $\mathcal{L}\left(H\right)$-valued p.d. kernels
on $S\times S$. In particular, $\mathscr{K}\left(S,\mathbb{C}\right)$
is the set of all scalar-valued p.d. kernels. $\mathcal{S}_{p}\left(H\right)$
denotes the $p$th Schatten-class operators, with norm 
\[
\left\Vert T\right\Vert _{p}=\left(\sum\nolimits_{n}s_{n}\left(T\right)^{p}\right)^{1/p}<\infty,
\]
where $s_{n}\left(T\right)$ are the singular numbers of $T$. Here,
$\mathcal{S}_{1}\left(H\right)$ is the trace class, and $\mathcal{S}_{2}\left(H\right)$
the class of Hilbert-Schmidt operators. $\mathcal{S}_{1}^{+}\left(H\right)$
is the cone of positive trace class operators, and $\mathcal{S}_{1}^{1}\left(H\right)=\left\{ T\in\mathcal{S}_{1}^{+}\left(H\right):tr\left(T\right)=1\right\} $
consists of quantum states. 
\begin{lem}
Fix $K\in\mathscr{K}\left(S,\mathcal{L}\left(H\right)\right)$. For
all $\rho\in\mathcal{S}_{1}^{+}\left(H\right)$, the kernel $c:S\times S\rightarrow\mathbb{C}$,
\[
c\left(s,t\right):=tr\left(\rho K\left(s,t\right)\right)
\]
is positive definite, that is, $c\in\mathscr{K}\left(S,\mathbb{C}\right)$. 
\end{lem}

\begin{proof}
Write 
\[
\rho=\sum\lambda_{i}\left|\varphi_{i}\left\rangle \right\langle \varphi_{i}\right|,\quad\lambda_{i}>0,
\]
where $\left(\varphi_{i}\right)$ is an ONB in $H$. Then, 
\begin{align*}
\left(s,t\right)\longmapsto tr\left(\rho K\left(s,t\right)\right) & =\sum\lambda_{i}\left\langle \varphi_{i},K\left(s,t\right)\varphi_{i}\right\rangle _{H}
\end{align*}
is a sum of scalar valued p.d. kernels, and so it is in $\mathscr{K}\left(S,\mathbb{C}\right)$. 
\end{proof}
\begin{thm}
\label{thm:d2}Fix $K\in\mathscr{K}\left(S,\mathcal{L}\left(H\right)\right)$.
Consider the map $\Psi:\mathcal{S}_{1}^{+}\left(H\right)\rightarrow\mathcal{K}\left(S,\mathbb{C}\right)$,
where 
\[
\Psi\left(\rho\right)\left(s,t\right)=tr\left(\rho K\left(s,t\right)\right),\quad\left(s,t\right)\in S\times S.
\]
Then $\Psi$ is 1-1 if and only if the double commutant $\left\{ K\left(s,t\right)\right\} ''$
is weak{*} dense in $\mathcal{L}\left(H\right)$. 
\end{thm}

\begin{proof}
This follows from von Neumann's double commutant theorem. Recall that
$\mathcal{S}_{1}^{+}\left(H\right)$ is the space of normal states,
and it is the predual of $\mathcal{L}\left(H\right)$. Therefore,
the injectivity of $\Psi$ in question is equivalent to the von Neumann
algebra generated by $\left\{ K\left(s,t\right):s,t\in S\right\} $
being weak{*} dense in $\mathcal{L}\left(H\right)$. 
\end{proof}
\begin{defn}
An $\mathcal{L}\left(H\right)$-valued p.d. kernel $K$ is said to
be complete if the map 
\[
\mathcal{S}_{1}^{+}\left(H\right)\ni\rho\longmapsto tr\left(\rho K\left(\cdot,\cdot\right)\right)\in\mathcal{K}\left(S,\mathbb{C}\right)
\]
is 1-1. That is, the scalar-valued kernel $\left(s,t\right)\mapsto tr\left(\rho K\left(s,t\right)\right)$
on $S\times S$ uniquely determines $K\in\mathscr{K}\left(S,\mathcal{L}\left(H\right)\right)$.
\end{defn}

\begin{cor}
\label{cor:d4}Fix $K\in\mathscr{K}\left(S,\mathcal{L}\left(H\right)\right)$.
Let $H_{\tilde{K}}$ be the RKHS of the corresponding scalar-valued
kernel $\tilde{K}$, defined on $\left(S\times H\right)\times\left(S\times H\right)$.
Let $\left\{ W_{t}\left(\cdot\right)\right\} _{t\in S}$ be the associated
$H$-valued Gaussian process. Assume that $K$ is complete. Then the
system of vectors 
\[
\left\{ V_{t}a=\tilde{K}\left(\cdot,\left(t,a\right)\right)\right\} _{t\in S}\subset H_{\tilde{K}}
\]
determines each $a\in H$ uniquely. Similarly, the scalar valued Gaussian
process 
\[
\left\{ \left\langle W_{t},a\right\rangle _{H}\right\} _{t\in S}\subset L^{2}\left(\Omega,\mathbb{P}\right)
\]
determines each $a\in H$ uniquely.
\end{cor}

\begin{proof}
Let $a\in H$. Apply \prettyref{thm:d2} to $\rho=\left|a\left\rangle \right\langle a\right|\in\mathcal{S}_{1}^{+}\left(H\right)$,
and use the fact that 
\[
\left\langle V_{s}a,V_{t}a\right\rangle _{H_{\tilde{K}}}=\mathbb{E}\left[\left\langle a,W_{s}\right\rangle _{H}\left\langle W_{t},a\right\rangle _{H}\right]=\left\langle a,K\left(s,t\right)a\right\rangle _{H}=tr\left(\rho K\left(s,t\right)\right).
\]
\end{proof}

\subsection{\label{subsec:5-1}An inverse problem}

In the theory of quantum measurement, a pertinent question is how
to recover the state of a system given measurement data. More precisely,
for generalized channels, one may consider the following: 
\begin{question}
Given a generalized channel $K\in\mathscr{K}\left(S,\mathcal{L}\left(H\right)\right)$,
and measurement data $c\left(s,t\right)=tr\left(\rho K\left(s,t\right)\right)$,
how to recover the quantum state $\rho\in\mathcal{S}_{1}^{1}\left(H\right)$? 

Alternatively, given $K$, $\left(s_{i}\right)_{i=1}^{N}\subset S$
and $\left(c_{i,j}\right)_{i,j=1}^{N}\subset\mathbb{C}$, consider
the following optimization problem: 
\begin{equation}
\min\left\{ \sum_{i,j=1}^{N}\left|tr\left(\rho K\left(s_{i},s_{j}\right)\right)-c_{i,j}\right|^{2}:\rho\in\mathcal{S}_{1}^{1}\left(H\right)\right\} .\label{eq:e1}
\end{equation}
\end{question}

Since every state $\rho\in\mathcal{S}_{1}^{+}\left(H\right)$ is a
product of Hilbert-Schmidt operators ($\rho=\left(\sqrt{\rho}\right)^{2}$),
we will first extend $\mathcal{L}\left(H\right)$-valued p.d. kernels
to the case where $H=\mathcal{S}_{2}$. 

In fact, every $\mathcal{L}\left(H\right)$-valued p.d. kernel is
also an $\mathcal{L}\left(\mathcal{S}_{2}\right)$-valued kernel.
Here, $\mathcal{S}_{2}=\mathcal{S}_{2}\left(H\right)$ is the class
of Hilbert-Schmidt operators in $H$, which is a Hilbert space with
inner product $\left\langle A,B\right\rangle _{2}=tr\left(A^{*}B\right)$.
More precisely, one has: 
\begin{prop}
$\mathscr{K}\left(S,\mathcal{L}\left(H\right)\right)\subset\mathscr{K}\left(S,\mathcal{L}\left(S_{2}\right)\right)$.
\end{prop}

\begin{proof}
Let $K\in\mathscr{K}\left(S,\mathcal{L}\left(H\right)\right)$. Note
that $K\in\mathscr{K}\left(S,\mathcal{L}\left(\mathcal{S}_{2}\right)\right)$
if and only if 
\[
\sum_{i,j}\left\langle A_{i},K\left(s_{i},s_{j}\right)A_{j}\right\rangle _{2}=\sum_{i,j}tr\left(A_{i}^{*}K\left(s_{i},s_{j}\right)A_{j}\right)\geq0,\quad\forall\left(A_{i}\right)_{i=1}^{N}\subset\mathcal{S}_{2}.
\]
It suffices to verify that 
\[
\sum_{i,j}\left\langle \varphi,A_{i}^{*}K\left(s_{i},s_{j}\right)A_{j}\varphi\right\rangle _{H}\geq0,\quad\forall\varphi\in H.
\]
This follows from the p.d. assumption of $K$, since 
\begin{align*}
\sum_{i,j}\left\langle \varphi,A_{i}^{*}K\left(s_{i},s_{j}\right)A_{j}\varphi\right\rangle _{H} & =\sum_{i,j}\left\langle A_{i}\varphi,K\left(s_{i},s_{j}\right)A_{j}\varphi\right\rangle _{H}\\
 & =\sum_{i,j}\left\langle h_{i},K\left(s_{i},s_{j}\right)h_{j}\right\rangle _{H}\geq0,\quad h_{i}:=A_{i}\varphi.
\end{align*}
\end{proof}
\begin{rem*}
Thus, there is an associated scalar-valued p.d. kernel $\tilde{K}:\left(S\times\mathcal{S}_{2}\right)\times\left(S\times\mathcal{S}_{2}\right)\rightarrow\mathbb{C}$,
with 
\[
\tilde{K}\left(\left(s,A\right),\left(t,B\right)\right)=\left\langle A,K\left(s,t\right)B\right\rangle _{2}.
\]
And, an $\mathcal{S}_{2}$-valued Gaussian process $W$ satisfying
that 
\[
\mathbb{E}\left[\left\langle A,W_{s}\right\rangle _{2}\left\langle W_{t},B\right\rangle _{2}\right]=\left\langle A,K\left(s,t\right)B\right\rangle _{2}
\]
for all $s,t\in S$, and $A,B\in\mathcal{S}_{2}$. 
\end{rem*}
Below, we analyze a simple case when $\dim H<\infty$, which reduces
\eqref{eq:e1} to a constraint kernel approximation (least-squares)
problem. 

Assume $H$ is finite-dimensional. Let 
\begin{equation}
K\left(s,t\right)=Q\left(s\cap t\right),\quad s,t\in\mathscr{B}_{\Omega},\label{eq:E2}
\end{equation}
where $Q:\left(\Omega,\mathscr{B}_{\Omega}\right)\rightarrow\mathcal{L}\left(H\right)$
is a POVM over a measurable space $\left(\Omega,\mathscr{B}_{\Omega}\right)$.
Recall that $Q\left(\Omega\right)=I_{H}$. Introduce $\tilde{K}$
and its RKHS $H_{\tilde{K}}$ as follows: 
\begin{align*}
\tilde{K}\left(\left(s,A\right),\left(t,B\right)\right) & =\left\langle \tilde{K}_{\left(s,A\right)},\tilde{K}_{\left(t,B\right)}\right\rangle _{H_{\tilde{K}}}\\
 & =\left\langle A,Q\left(s\cap t\right)B\right\rangle _{2}=tr\left(A^{*}Q\left(s\cap t\right)B\right),
\end{align*}
for all Borel sets $s,t\in\mathscr{B}_{\Omega}$, and all $A,B\in\mathcal{S}_{2}\left(H\right)$
(i.e., all $A,B\in\mathcal{L}\left(H\right)$, since $\dim H<\infty$). 

Set $V:\mathcal{S}_{2}\left(H\right)\rightarrow H_{\tilde{K}}$ by
\[
VA=\tilde{K}\left(\cdot,\left(\Omega,A\right)\right).
\]
This is an isometric embedding of $\mathcal{S}_{2}\left(H\right)$
into $H_{\tilde{K}}$, since 
\[
\left\Vert VA\right\Vert _{H_{\tilde{K}}}^{2}=\left\langle \tilde{K}\left(\cdot,\left(\Omega,A\right)\right),\tilde{K}\left(\cdot,\left(\Omega,A\right)\right)\right\rangle _{H_{\tilde{K}}}=\left\langle A,Q\left(\Omega\right)A\right\rangle _{2}=\left\Vert A\right\Vert _{2}^{2}.
\]

\begin{lem}
Given $\left(c_{i}\right)_{i=1}^{N}$ in $\mathbb{R}_{+}$, and $\left(s_{i}\right)_{i=1}^{N}$
in $\mathscr{B}_{\Omega}$, we have 
\[
\sum_{i}\left|tr\left(AQ\left(s_{i}\right)\right)-c_{i}\right|^{2}=\sum_{i}\left|\left\langle \tilde{K}_{\left(s_{i},I_{H}\right)},\tilde{K}_{\left(\Omega,A\right)}\right\rangle _{H_{\tilde{K}}}-c_{i}\right|^{2}.
\]
\end{lem}

\begin{proof}
Note that 
\begin{align*}
tr\left(AQ\left(s\right)\right) & =tr\left(I_{H}^{*}Q\left(s\right)A\right)\\
 & =\left\langle I_{H},Q\left(s\right)A\right\rangle _{2}=\left\langle \tilde{K}_{\left(s,I_{H}\right)},\tilde{K}_{\left(\Omega,A\right)}\right\rangle _{H_{\tilde{K}}}.
\end{align*}
\end{proof}
\begin{cor}
In the current setting, \eqref{eq:e1} is converted to the following
constraint kernel approximation problem: 
\begin{gather}
\min\left\{ \sum_{i=1}^{N}\left|\left\langle \tilde{K}_{\left(s_{i},I_{H}\right)},\tilde{K}_{\left(\Omega,A\right)}\right\rangle _{H_{\tilde{K}}}-c_{i}\right|^{2}:A\in\mathcal{S}_{1}^{1}\left(H\right)\right\} \nonumber \\
\Updownarrow\nonumber \\
\min\left\{ \sum_{i=1}^{N}\left|F\left(s_{i},I_{H}\right)-c_{i}\right|^{2}:F\in V\left(\mathcal{S}_{1}^{1}\left(H\right)\right)\subsetneq H_{\tilde{K}}\right\} .\label{eq:5-3}
\end{gather}
\end{cor}

\begin{proof}
Recall that $VA\left(\cdot\right)=\tilde{K}\left(\cdot,\left(\Omega,A\right)\right)$,
so that 
\[
\left\langle \tilde{K}_{\left(s_{i},I_{H}\right)},\tilde{K}_{\left(\Omega,A\right)}\right\rangle _{H_{\tilde{K}}}=\left(VA\right)\left(s_{i},I_{H}\right)
\]
by the reproducing property of $H_{\tilde{K}}$. Then $F$ in \eqref{eq:5-3}
is of the form $F=VA$, $A\in\mathcal{S}_{1}^{1}\left(H\right)$. 
\end{proof}
\begin{rem}
Kernel approximation methods are widely used in statistical learning
for tasks such as regression, classification, and density estimation.
They are particularly useful in settings where the underlying functional
relationships are complex and non-linear, and where a balance between
fitting the training data and maintaining generalizability to new
data is crucial. The standard formulation is as follows: 
\begin{equation}
\min\left\{ \sum_{i=1}^{N}\left|f\left(x_{i}\right)-y_{i}\right|^{2}+\lambda\left\Vert f\right\Vert _{H_{K}}^{2}:f\in H_{K}\right\} \label{eq:e3}
\end{equation}
where $H_{K}$ is the RKHS of a scalar-valued kernel $K$ on some
set $S\times S$, $\left(y_{i}\right)\subset\mathbb{R}$, and $\lambda$
is a regularization parameter. 

A key difference between \eqref{eq:5-3} and \eqref{eq:e3} is that,
in \eqref{eq:5-3}, $F$ is restricted to a convex subset, rather
than being defined over the entire space $H_{\tilde{K}}$. A standard
method to solve \eqref{eq:5-3} is to use projected gradient descent.
\end{rem}

\begin{rem}
For a given POVM $Q$ as in \eqref{eq:E2}, the construction of $H_{\tilde{K}}$
and $V$ in the above example can be used to show that $Q$ has a
PVM dilation in $H_{\tilde{K}}$. The verification is similar to that
of the Stinespring dilation theorem in \prettyref{sec:3} (also see
\cite[Section 2]{MR4787903}).
\end{rem}

\begin{acknowledgement*}
We thank the anonymous referee for pointing out that the construction
of $H_{\tilde{K}}$ underlying our approach was originally developed
in \cite{MR2938971} and for bringing Pedrick\textquoteright s \textquotedblleft Tilde
Correspondence\textquotedblright{} to our attention \cite{MR4250453}.
\end{acknowledgement*}
\bibliographystyle{amsalpha}
\nocite{*}
\bibliography{ref}

\end{document}